\documentclass{article}

\usepackage{amsmath,amssymb,amsthm}

\usepackage{tikz}
\usetikzlibrary{external}

\usetikzlibrary{patterns}
\usetikzlibrary{plotmarks}
\usetikzlibrary{shapes}
\tikzstyle{legend}=[draw,fill=white,text width=12em,text ragged,
anchor=north east,inner sep=2pt]
\tikzstyle{legendsw}=[draw,fill=white,text width=12em,text ragged,
anchor=south west,inner sep=2pt]
\colorlet{dimgray}{black!5}
\colorlet{lightgray}{black!20}

\newtheorem{theorem}{Theorem}[section]

\newtheorem{lemma}[theorem]{Lemma}
\newtheorem{proposition}{Proposition}

\theoremstyle{definition}

\newtheorem*{IDT}{Interpolation Decoding Algorithm}

\newcommand{\F}{\mathbb{F}}

\newcommand{\A}{\mathbb{A}}
\newcommand{\ga}{\alpha}
\newcommand{\gb}{\beta}

\newcommand{\gd}{\delta}

\newcommand{\gl}{\lambda}

\newcommand{\gw}{\omega}
\renewcommand{\phi}{\varphi}

\newcommand{\calP}{\mathcal{P}}

\newcommand{\frakm}{\mathfrak{m}}

\newcommand{\set}[1]{\{#1\}}
\newcommand{\tuple}[1]{\langle{#1}\rangle}

\newcommand{\ideal}[1]{\langle{#1}\rangle}
\newcommand{\ev}{\mathrm{ev}}
\newcommand{\wt}{\mathrm{wt}}

\newcommand{\LT}{\mathrm{lt}}
\newcommand{\LM}{\mathrm{lm}}
\newcommand{\LC}{\mathrm{lc}}
\newcommand{\spoly}{\mathrm{spoly}}

\begin{document}

\title{Unique Decoding of Plane AG Codes Revisited}

\author{Kwankyu Lee\thanks{The author is the Department of Mathematics, Chosun University, Gwangju 501-759, Korea. This work was supported by Basic Science Research Program through the National Research Foundation of Korea(NRF) funded by the Ministry of Education, Science and Technology(2009-0064770) and also by research fund from Chosun University, 2008.}}

\maketitle

\begin{abstract}
We reformulate a recently introduced interpolation-based unique decoding algorithm of algebraic geometry codes using the theory of Gr\"obner bases of modules on the coordinate ring of the base curve. With the same decoding performance, the new algorithm has a more conceptual description that lets us better understand the majority voting procedure central in the interpolation-based unique decoding.  
\end{abstract}

\section{Introduction}

Recently a new kind of unique decoding algorithm of algebraic geometry codes appeared \cite{kwankyu2011}. The algorithm decodes the primal AG code that consists of codewords obtained by evaluation of functions at rational points of an algebraic curve, unlike the classical syndrome decoding algorithm that decodes the dual code. Based on Gr\"obner bases of modules over a univariate polynomial ring, the algorithm has a regular data and control structure that is suitable for parallel hardware implementation, like K\"otter's algorithm for the syndrome decoding \cite{koetter1998}. 

In this paper, we reformulate the previous algorithm, using the theory of Gr\"obner bases of modules on the coordinate ring of the base curve. This approach eliminates the technical complexity of the previous algorithm in a large degree, and results in a conceptually clean description of the algorithm which would contribute for better understanding the majority voting procedure, which plays a central role in the interpolation-based unique decoding. The new algorithm resembles the Berlekamp-Massey-Sakata algorithm for the syndrome decoding \cite{sakata1995}. 

In Section \ref{mckwd}, we review the theory of the Gr\"obner bases of modules over the coordinate rings of algebraic curves, and outline the interpolation decoding algorithm based on it.  The algorithm operates by iterating two core steps, the Gr\"obner basis computation step and the message guessing step. Sections \ref{jksadde} and \ref{kleflso} are devoted for each step. In Section \ref{ckdwdd}, we demonstrate the algorithm with Hermitian codes. In the remainder of this introduction, we briefly review basic facts about AG codes. Like the previous algorithm in \cite{kwankyu2011} and the BMS decoding algorithm, the new algorithm is formulated for the AG codes from the Miura-Kamiya curves \cite{miura2007}, which include Hermitian curves as prominent special cases.

A Miura-Kamiya curve $X$ is an irreducible plane curve defined by the equation 
\[
	y^a+\sum_{ai+bj<ab}c_{i,j}x^iy^j+dx^b=0
\]
over a field $\F$ with $\gcd(a,b)=1$ and $0\neq d\in\F$.  It is well known that $X$ has a unique point $P_\infty$ at infinity and has a unique valuation $v_{P_\infty}$ associated with it. Let $\gd(f)=-v_{P_\infty}(f)$ for $f$ in the coordinate ring $R$ of $X$. Then $\gd(x)=a$ and $\gd(y)=b$. By the equation of the curve, a function in the coordinate ring $R=\F[x,y]$ can be written as a unique $\F$-linear combination of monomials $x^iy^j$ with $i\ge 0$ and $0\le j<a$, which we call monomials of $R$. The numerical semigroup of $R$ at $P_\infty$,
\[
	\begin{split}
	S&=\set{\gd(f)\mid f\in R}=\set{\gd(x^iy^j)\mid i\ge 0,0\le j<a}\\
	&=\set{ai+bj\mid i\ge 0,0\le j<a}=\ideal{a,b}
	\end{split}
\]
is a subset of the Weierstrass semigroup at $P_\infty$. As $\gcd(a,b)=1$, there is an integer $b'$ such that $b'b\equiv 1\pmod{a}$. If $s=ai+bj$ is a nongap, then $b's\bmod a=j$, $(s-bj)/a=i$, and therefore $i$ and $j$ are uniquely determined. Hence the monomials of $R$ are in one-to-one correspondence with nongaps in $S$. For a nongap $s$, let $\phi_s$ be the unique monomial with $\gd(\phi_s)=s$. 

Let $\calP=\set{P_1,P_2,\dots,P_n}$ be a set of nonsingular rational points of $X$. The evaluation $\ev$ from $R$ to the Hamming space $\F^n$ defined by 
\[
	\phi\mapsto(\phi(P_1),\phi(P_2),\dots,\phi(P_n))
\]
is a linear map over $\F$. Let $u$ be a fixed positive integer less than $n$ and define
\[
	L_u=\set{f\in R\mid \gd(f)\le u}=\langle\phi_s\mid s\in S, s\le u\rangle,
\]
where brackets denote the linear span over $\F$. Then the AG code $C_u$ is defined as the image of $L_u$ under $\ev$. As $u<n$, the evaluation is one-to-one on $L_u$. Therefore the dimension of the linear code $C_u$ equals $\dim_\F L_u=|\set{s\in S\mid s\le u}|$. 

\section{Interpolation decoding}\label{mckwd}

We assume a codeword $c$ in $C_u$ is sent through a noisy communication channel and $v\in\F^n$ is the vector  received from the channel. Let $v=c+e$ with the error vector $e$. Then $c=\ev(\mu)$ for a unique 
\[
	\mu=\sum_{s\in S,s\le u}\gw_s\phi_s\in L_u,\qquad \gw_s\in\F
\]
We assume encoding by evaluation, and the vector $(\gw_s\mid s\in S,s\le u)$ is the message encoded into the codeword $c$. The decoding problem is essentially to find $\gw_s$ for all nongap $s\le u$ from the given $v$.

For $s\ge u$, let $v^{(s)}=v$, $c^{(s)}=c$, and $\mu^{(s)}=\mu$. For nongap $s\le u$, let
\[
\begin{aligned}
	\mu^{(s-1)}&=\mu^{(s)}-\gw_s\phi_s,\\
	c^{(s-1)}&=c^{(s)}-\ev(\gw_s\phi_s),\\
	v^{(s-1)}&=v^{(s)}-\ev(\gw_s\phi_s),
\end{aligned}
\]
and for gap $s\le u$, let $v^{(s-1)}=v^{(s)}$, $c^{(s-1)}=c^{(s)}$, and $\mu^{(s-1)}=\mu^{(s)}$. Note that 
\[
	\mu^{(s)}\in L_s,\quad c^{(s)}=\ev(\mu^{(s)})\in C_s,\quad v^{(s)}=c^{(s)}+e
\]
for all $s$. Hence we can find $\gw_s$ iteratively.

A polynomial in $R[z]$ defines a function on the product surface of $X$ and the line $\A_\F^1$, and can be evaluated at a point $(P,\ga)$ with $P\in X,\ga\in\F$. Hence we can define the \emph{interpolation module}
\[
	I_v=\set{ f \in Rz\oplus R\mid f(P_i,v_i)=0, 1\le i\le n}
\]
for $v$ and similarly for $v^{(s)}$. These interpolation modules are indeed modules over $R$, and finite-dimensional vector space over $\F$. Note that
\begin{equation}\label{ckskw}
	I_v=R(z-h_v)+J
\end{equation}
where
\[
	J=\bigcap_{1\le i\le n}\frakm_i,
	\qquad
	\ev(h_v)=v,
\]
and $\frakm_i=\ideal{x-\ga_i,y-\gb_i}$ is the maximal ideal of $R$ associated with $P_i=(\ga_i,\gb_i)$. Recall that by Lagrange interpolation, $h_v$ can be computed fast from $v$. We will see that the key to find $\gw_s$ is the Gr\"obner basis of $I_{v^{(s)}}$ with respect to a monomial order $>_s$, which is defined in the following. 

Let $s$ be an integer. The monomial $x^iy^jz^k$  of $R[z]$ is given the weight $\gd(x^iy^j)+sk$. In particular, the weighted degrees of the monomials $x^iy^jz$ and $x^iy^j$ of $Rz\oplus R$ are $ai+bj+s$ and $ai+bj$, respectively. The monomial order $>_s$ on $Rz\oplus R$ orders the monomials by their weighted degrees, and breaks the tie with higher $z$-degree. For $f$ in $Rz\oplus R$, the notations $\LT_s(f)$, $\LM_s(f)$, and $\LC_s(f)$ denote the leading term, the leading monomial, and the leading coefficient of $f$, respectively, with respect to $>_s$. As $f=f^Uz+f^D$ with unique $f^U,f^D\in R$\footnote{The superscripts $U$ and $D$ may be read ``upstairs'' and ``downstairs'', respectively (with $z$ being the staircase).}, note that
\[
	\LM_s(f)\in Rz \iff \gd(f^U)+s\ge \gd(f^D),
\]
where equality holds if and only if $\LM_s(f)\in Rz$ and $\LM_{s-1}(f)\in R$.
 
Now let $M$ be a submodule of $Rz\oplus R$. A subset $B$ of $M$ is called a \emph{Gr\"obner basis} with respect to $>_s$ if the leading term of every element of $M$ is divided by the leading term of some element of $B$. We will write
 \[
	B=\set{G_i,F_j}
\]
with $i$ and $j$ in some implicit index sets, where the leading term of $G_i$ is in $R$ while that of $F_j$ is in $Rz$. The \emph{sigma set} $\Sigma_s$ or $\Sigma_s(M)$ of $M$ is the set of all leading monomials of polynomials in $M$ with respect to $>_s$. The \emph{delta set} $\Delta_s$ or $\Delta_s(M)$ of $M$ is the complement of $\Sigma_s$ in the set of all monomials of $Rz\oplus R$. We note that
\[
	\begin{aligned}
	\Delta_s&=(Rz\cap\Delta_s)\sqcup(R\cap\Delta_s),\\
	\Sigma_s&=(Rz\cap\Sigma_s)\sqcup(R\cap\Sigma_s).
	\end{aligned}
\]
where $\sqcup$ denotes disjoint union. For the case that $M$ is an ideal of $R$, we may omit the superfluous $s$ from the above notations, and denote $>_s$ simply by $>$ in particular. Note that if $\LM_s(f)\in Rz$, then $\LM_s(f)=\LM(f^U)z$, and if $\LM_s(f)\in R$, then $\LM_s(f)=\LM(f^D)$. It is easy to see by the definition of Gr\"obner bases that
\[
	\begin{split}
	\dim_\F(Rz\oplus R/M)&=|\Delta_s|=|\Delta_s\cap Rz|+|\Delta_s\cap R|\\
	&=|\Delta(\set{F_j^U})|+|\Delta(\set{G_i^D})|,
	\end{split}
\]
where $\Sigma(T)$, $\Delta(T)$ with a set of polynomials in $R$ have natural definitions. 

As $J$ is an ideal of $R$, it has a Gr\"obner basis $\set{\eta_i}$ with respect to $>$, and 
\begin{equation}\label{camdzs}
	\dim_\F R/J=|\Delta(J)|=|\Delta(\set{\eta_i})|=n
\end{equation}
since $J$ is the ideal associated with the sum of $n$ rational points on $X$. By \eqref{ckskw}, we see that $\dim_\F(Rz\oplus R/I_v)=\dim_\F(R/J)=n$. Let $N=\gd(h_v)$. The set $\set{\eta_i}\cup\set{z-h_v}$ is then a Gr\"obner basis of $I_v$ with respect to $>_N$. Let us denote a Gr\"obner basis of $I_{v^{(s)}}$ with respect to $>_s$ by $B^{(s)}=\set{G_i,F_j}$. Observe that if $s$ is a nongap $\le u$, then the set $\tilde{B}=\set{G_{i}(z+\gw_s\phi_s),F_j(z+\gw_s\phi_s)}$ is still a Gr\"obner basis of $I_{v^{(s-1)}}$ with respect to $>_s$, but not with respect to $>_{s-1}$ in general. These observations lead to the following interpolation decoding algorithm. 
 
\begin{IDT}
Let $v$ be the received vector.
\begin{description}
\item[Initialize]
Compute $h_v$. Let $B^{(N)}=\set{\eta_i}\cup\set{z-h_v}$ where $N=\gd(h_v)$.
\item[Main]
Repeat the following for $s$ from $N$ to $0$. 
\begin{description}
\item[M1] 
If $s$ is a nongap $\le u$, then make a guess $w^{(s)}$ for $\gw_s$, and let $\tilde{B}=\set{G_{i}(z+w^{(s)}\phi_s),F_j(z+w^{(s)}\phi_s)}$. Otherwise, let $\tilde{B}=B^{(s)}$.
\item[M2]
Compute $B^{(s-1)}$ from $\tilde{B}$.
\end{description}
\item[Finalize]
Output $(w^{(s)}\mid \text{nongap $s\le u$})$.
\end{description}
\end{IDT}

In the next section, we will elaborate on the step \textbf{M2}. The results in the section will lay a foundation for Section \ref{kleflso}, in which we give details of the main steps \textbf{M1} and \textbf{M2}.  

\section{Gr\"obner basis computation}\label{jksadde}

First we review the concept of the \emph{lcm}, least common multiple, for the monomials of $R$. For two monomials $\phi_s$ and $\phi_t$, we say $\phi_s$ \emph{divides} $\phi_t$ if there exists a unique monomial $\gl$ such that
\[
	\gd(\phi_t-\gl\phi_s)<\gd(\phi_t).
\]
The unique monomial $\gl$ will be denoted by the quotient $\phi_t/\phi_s$. Note that $\phi_s$ divides $\phi_t$ if and only if $t-s$ is a nongap, and in this case, actually $\gl=\phi_{t-s}$. We will also simply say $s$ \emph{divides} $t$ if $t-s$ is a nongap.

\begin{proposition}
Let $s$ and $t$ be nongaps that do not divide each other. Then there are unique nongaps $l_1$ and $l_2$ such that
$l_1$ and $l_2$ are both divisible by $s$ and $t$, and if a nongap $c$ is divisible by $s$ and $t$, then $l_1$ or $l_2$ divides $c$.
\end{proposition}

\begin{proof}
Let $s=as_1+bs_2$ and $t=at_1+bt_2$. Without loss of generality, we may assume $s_1<t_1$ and $s_2>t_2$. Since $s$ divides $a(s_1+b)$, we also have $t_1<s_1+b$. 

Now let $l_1=at_1+bs_2$ and $l_2=a(s_1+b)+bt_2$. It is easily verified that $l_1$ and $l_2$ are divisible by $s$ and $t$. Suppose a nongap $c=ac_1+bc_2$ is divisible by $s$ and $t$. Then $c-s=a(c_1-s_1)+b(c_2-s_2)$ is a nongap as $s$ divides $c$. Note that
\[
	b'(c-s)\bmod a=\begin{cases}
	c_2-s_2 & \text{if $c_2\ge s_2$,}\\
	c_2-s_2+a & \text{if $c_2<s_2$.}
	\end{cases}
\]
Therefore if $c_2\ge s_2$, then $c_1\ge s_1$ while if $c_2<s_2$, then $c_1\ge s_1+b$. Similarly, as $t$ divides $c$, if $c_2\ge t_2$, then $c_1\ge t_1$ while if $c_2<t_2$, then $c_1\ge t_1+b$. So in any case, we have at least $c_1\ge t_1$. Now let us check that $c$ is divisible either by $l_1$ or $l_2$. Assume $l_1$ does not divide $c$. Then $c_1<t_1+b$ as $c_1\ge t_1+b$ contradicts our assumption. Therefore $c_2\ge t_2$. If $c_1<s_1+b$, then $c_2\ge s_2$, which also contradicts the assumption. Therefore $c_1\ge s_1+b$. Then $l_2$ divides $c$. 
\end{proof}

We will call $\phi_{l_1}$ and $\phi_{l_2}$ the lcms of $\phi_s$ and $\phi_t$. In the case when $\phi_s$ divides $\phi_t$, we will call $\phi_t$ the lcm of $\phi_s$ and $\phi_t$. 

Let $B=\set{G_i,F_j}$ be a Gr\"obner basis of a submodule $M$ of $Rz\oplus R$ with respect to $>_s$. We want to compute a Gr\"obner basis of the same module $M$ with respect to $>_{s-1}$ from $B$. Note that while $\LM_{s-1}(G_i)=\LM_s(G_i)\in R$, we may have either $\LM_{s-1}(F_j)=\LM_s(F_j)\in Rz$ or $\LM_{s-1}(F_j)\in R$. Let $\Sigma_s$ and $\Delta_s$ denote the sigma set and the delta set of $M$ with respect to $>_s$, respectively. Observe that 
\[
	\begin{gathered}
	Rz\cap\Sigma_{s-1} \subset Rz\cap\Sigma_s,\quad 
	R\cap\Sigma_{s-1} \supset R\cap\Sigma_s,\\
	Rz\cap\Delta_{s-1} \supset Rz\cap\Delta_s,\quad 
	R\cap\Delta_{s-1} \subset R\cap\Delta_s.
	\end{gathered}
\]
For those $j$ such that $\LM_{s-1}(F_j)=\LM_s(F_j)\in Rz$, define 
\[
	\spoly(F_j)=\set{F_j}.
\]
If $\LM_{s-1}(F_j)\in R\cap\Sigma_s$, then there is an $i$ such that $\LM_s(G_i)|\LM_{s-1}(F_j)$, and then, with one such $i$, define
\[
	\spoly(F_j)=\set{\frac{1}{\LC_{s-1}(F_j)}F_j-\frac{\LM_{s-1}(F_j)}{\LT_s(G_i)}G_i}.
\]
Finally, if $\LM_{s-1}(F_j)\in R\cap\Delta_s$, then define
\[
	\spoly(F_j)=\set{\frac{\psi}{\LT_{s-1}(F_j)}F_j-\frac{\psi}{\LT_s(G_i)}G_i\mid\text{$\psi$ is an lcm of $\LM_{s-1}(F_j)$ and $\LM_s(G_i)$}}.
\]

\begin{proposition}\label{kcmald}
For every $f\in\spoly(F_j)$, $\LM_{s-1}(f)$ is in $Rz$.
\end{proposition}

\begin{proof}
Recall that $\LM_s(F_j)\in Rz$. Suppose $\LM_{s-1}(F_j)\in R$, and let $\psi$ be an lcm of $\LM_{s-1}(F_j)$ and $\LM_s(G_i)$ for any $i$. Then
\[
	\begin{gathered}
	\gd(\frac{\psi}{\LT_{s-1}(F_j)}F_j^U)=\gd(\psi)-\gd(F_j^D)+\gd(F_j^U)=\gd(\psi)-s,\\
	\gd(\frac{\psi}{\LT_s(G_i)}G_i^U)=\gd(\psi)-\gd(G_i^D)+\gd(G_i^U)<\gd(\psi)-s.
	\end{gathered}	
\]
Therefore
\[
	\gd((\frac{\psi}{\LT_{s-1}(F_j)}F_j-\frac{\psi}{\LT_s(G_i)}G_i)^U)=\gd(\psi)-s.
\]
On the other hand,
\[
	\begin{gathered}
	\gd(\frac{\psi}{\LT_{s-1}(F_j)}F_j^D)=\gd(\psi)-\gd(F_j^D)+\gd(F_j^D)=\gd(\psi),\\
	\gd(\frac{\psi}{\LT_s(G_i)}G_i^D)=\gd(\psi)-\gd(G_i^D)+\gd(G_i^D)=\gd(\psi).
	\end{gathered}	
\]
As the monic terms cancel each other, we have
\[
	\gd((\frac{\psi}{\LT_{s-1}(F_j)}F_j-\frac{\psi}{\LT_s(G_i)}G_i)^D)<\gd(\psi).
\]
Therefore
\[
	\gd((\frac{\psi}{\LT_{s-1}(F_j)}F_j-\frac{\psi}{\LT_s(G_i)}G_i)^U)+s-1\ge\gd((\frac{\psi}{\LT_{s-1}(F_j)}F_j-\frac{\psi}{\LT_s(G_i)}G_i)^D),
\]
and hence
\begin{equation}\label{jfjwa}
	\LM_{s-1}(\frac{\psi}{\LT_{s-1}(F_j)}F_j-\frac{\psi}{\LT_s(G_i)}G_i)=\LM(\frac{\psi}{\LT_{s-1}(F_j)}F_j^U)z\in Rz.
\end{equation}
For the case when $\LM_{s-1}(F_j)\in R\cap\Sigma_s$, notice that $\LM_{s-1}(F_j)$ is the lcm.
\end{proof}

\begin{proposition}\label{qhkxa}
A monomial $\phi$ is in $R\cap\Sigma_{s-1}$ if and only if there exists an $i$ such that $\LM_{s-1}(G_i)|\phi$ or there exists a $j$ such that $\LM_{s-1}(F_j)\in R\cap\Delta_s$ and $\LM_{s-1}(F_j)|\phi$.
\end{proposition}

\begin{proof}
Both $\LM_{s-1}(G_i)|\phi$ and $\LM_{s-1}(F_j)|\phi$ imply $\phi\in R\cap\Sigma_{s-1}$. Let us show the converse. If $\phi\in R\cap\Sigma_s$, then $\LM_{s}(G_i)|\phi$ for some $i$, and therefore $\LM_{s-1}(G_i)|\phi$. 
As $R\cap\Sigma_{s-1}\supset R\cap \Sigma_s$, it remains to consider the case when $\phi\in R\cap (\Sigma_{s-1}\backslash\Sigma_s)$. 

Suppose $f\in M$ is such that $\phi=\LM_{s-1}(f)\in R\cap (\Sigma_{s-1}\backslash\Sigma_s)$. Since $\phi\notin R\cap\Sigma_s$, we must have $\LM_s(f)\in Rz$, and hence 
\[
	\gd(f^U)+s=\gd(f^D)=\gd(\phi).
\]
Then $\LM_s(F_j)|\LM_s(f)$ for some $j$. As $\LM_s(F_j)\in Rz$, we have $\gd(F_j^U)+s \ge \gd(F_j^D)$, where actually equality holds as we will show now. Assume the contrary, that is,
\[
	\gd(F_j^U)+s>\gd(F_j^D). 
\]
Then 
\[
	\begin{gathered}
	\gd(\frac{\LT_s(f)}{\LT_s(F_j)}F_j^D)=\gd(f^U)-\gd(F_j^U)+\gd(F_j^D)<\gd(f^U)+s=\gd(f^D),\\
	\gd(\frac{\LT_s(f)}{\LT_s(F_j)}F_j^U)=\gd(f^U)-\gd(F_j^U)+\gd(F_j^U)=\gd(f^U).
	\end{gathered}
\]
These imply
\[
	\LM_s(f-\frac{\LT_s(f)}{\LT_s(F_j)}F)=\LM(f^D)=\LM_{s-1}(f)=\phi,
\]
contradictory to the assumption $\phi\notin R\cap\Sigma_s$. Hence $\gd(F_j^U)+s=\gd(F_j^D)$, and
\[
	\gd(\frac{\LM_s(f)}{\LM_s(F_j)}\LM_{s-1}(F_j))
	=\gd(f^U)-\gd(F_j^U)+\gd(F_j^D)=\gd(f^U)+s=\gd(\phi).
\]
Therefore $\LM_{s-1}(F_j)|\phi$, and $\LM_{s-1}(F_j)\in R\cap\Delta_s$.
\end{proof}

\begin{proposition}
A monomial $\phi$ is in $Rz\cap\Sigma_{s-1}$ if and only if there exists a $j$ such that $\LM_{s-1}(f)|\phi$ for some $f\in\spoly(F_j)$.
\end{proposition}

\begin{proof}
By Proposition \ref{kcmald}, the converse is clear. Let us assume $\phi\in Rz\cap\Sigma_{s-1}$. Suppose $\phi=\LM_{s-1}(f)$ for some $f\in M$. Then $\phi=\LM_s(f)$, and there exists some $j$ such that $\LM_s(F_j)|\phi$. 
If $\LM_{s-1}(F_j)\in Rz$, then $F_j\in\spoly(F_j)$ and $\LM_{s-1}(F_j)=\LM_s(F_j)|\phi$. 

Suppose $\LM_{s-1}(F_j)\in R\cap\Sigma_s$. Then there is an $i$ such that $\LM_s(G_i)|\LM_{s-1}(F_j)$ and 
\[
	\frac{1}{\LC_{s-1}(F_j)}F_j-\frac{\LM_{s-1}(F_j)}{\LT_s(G_i)}G_i\in\spoly(F_j)
\]
and by \eqref{jfjwa},
\[
	\LM_{s-1}(\frac{1}{\LC_{s-1}(F_j)}F_j-\frac{\LM_{s-1}(F_j)}{\LT_s(G_i)}G_i)=\LM_s(F_j)|\phi.
\]

Suppose $\LM_{s-1}(F_j)\in R\cap\Delta_s$. Note that
\[
	\gd(f^U)+s>\gd(f^D),\quad \gd(F_j^U)+s=\gd(F_j^D),
\]
and hence
\[
	\begin{gathered}
	\gd(\frac{\LT_s(f)}{\LT_s(F_j)}F_j^U)=\gd(f^U)-\gd(F_j^U)+\gd(F_j^U)=\gd(f^U),\\
	\gd(\frac{\LT_s(f)}{\LT_s(F_j)}F_j^D)=\gd(f^U)-\gd(F_j^U)+\gd(F_j^D)=\gd(f^U)+s>\gd(f^D).
	\end{gathered}	
\]
Thus we see that
\[
	\LM_s(f-\frac{\LT_s(f)}{\LT_s(F_j)}F_j)=\frac{\LM_s(f)}{\LM_s(F_j)}\LM_{s-1}(F_j)\in R
\]
and hence there is an $i$ such that 
\[
	\LM_s(G_i)|\frac{\LM_s(f)}{\LM_s(F_j)}\LM_{s-1}(F_j).
\]
Now there is an lcm $\psi$ of $\LM_{s-1}(F_j)$ and $\LM_s(G_i)$ such that
\begin{equation}\label{dkwda}
	\psi | \frac{\LM_s(f)}{\LM_s(F_j)}\LM_{s-1}(F_j),
\end{equation}
and
\[
	\frac{\psi}{\LT_{s-1}(F_j)}F_j-\frac{\psi}{\LT_s(G_i)}G_i\in\spoly(F_j).
\]
By \eqref{jfjwa},
\[
	\LM_{s-1}(\frac{\psi}{\LT_{s-1}(F_j)}F_j-\frac{\psi}{\LT_s(G_i)}G_i)
	=\frac{\psi}{\LM_{s-1}(F_j)}\LM_s(F_j)\in Rz
\]
and finally from \eqref{dkwda},
\[
	\frac{\psi}{\LM_{s-1}(F_j)}\LM_s(F_j)|\LM_s(f)=\phi.
\]
\end{proof}

Combining the above results, we see that the set
\[
\set{G_i,F_j\mid \LM_{s-1}(F_j)\in R\cap\Delta_s}\cup\bigcup_{j}\spoly(F_j)
\]
is a Gr\"obner basis of $M$ with respect to $>_{s-1}$. In general, the Gr\"obner basis may contain more polynomials than necessary. Indeed, we can reduce each set in the union by removing polynomials whose leading term is divisible by that of other polynomial in the same set. We will denote the \emph{reduced} Gr\"obner basis of $M$ with respect to $>_{s-1}$ by
 \[
\set{G_i,F_j\mid \LM_{s-1}(F_j)\in R\cap\Delta_s}'\cup\bigcup_{j}'\spoly(F_j).
\]

\section{Message Guessing}\label{kleflso}

The ideal of the error vector $e$ defined by
\[
	J_e=\bigcap_{e_i\neq 0}\frakm_i
\]
has a Gr\"obner basis $\set{\epsilon_i}$ with respect to $>$, and 
\begin{equation}\label{mdcad}
	\dim_\F R/J_e=|\Delta(J_e)|=\wt(e).
\end{equation}
Recall that $B^{(s)}=\set{G_i,F_j}$ is a Gr\"obner basis of $I_{v^{(s)}}$ with respect to $>_s$. Observe that $J_e(z-\mu^{(s)})\subset I_{v^{(s)}}$, which result in $\Sigma(J_e)z\subset\Sigma_s(I_{v^{(s)}})\cap Rz$, and hence $\Delta_s(I_{v^{(s)}})\cap Rz\subset\Delta(J_e)z$. Therefore
\[
	|\Delta_s(I_{v^{(s)}})\cap Rz|=|\Delta(F_j^U)|\le\wt(e).
\]
Now let $s$ be a nongap $\le u$. Let us consider the module
\[
	\tilde{I}_w=\set{f(z+w\phi_s)\mid f\in I_{v^{(s)}}}\subset Rz\oplus R.
\]
for $w\in\F$. Note that
\[
	\tilde{B}=\set{G_i(z+w\phi_s),F_j(z+w\phi_s)}
\]
is a Gr\"obner basis of $\tilde{I}_w$ with respect to $>_s$ since $\LM_s(f(z+w\phi_s))=\LM_s(f)$ for all $f\in I_{v^{(s)}}$. For the same reason, 
\[
	\Sigma_s(\tilde{I}_w)=\Sigma_s(I_{v^{(s)}}),\quad \Delta_s(\tilde{I}_w)=\Delta_s(I_{v^{(s)}}).
\]
Observe that $\tilde{I}_{\gw_s}=I_{v^{(s-1)}}$. Hence 
\begin{equation}\label{xjakd}
	|\Delta_{s-1}(\tilde{I}_{\gw_s})\cap Rz|\le\wt(e).
\end{equation}

In Theorem \ref{mxmjw} below, we will see that $\gw_s$ is such a $w$ that makes the value
\[
	|\Delta_{s-1}(\tilde{I}_w)\cap Rz|
\]
smallest, provided that $\wt(e)$ is not too large. First note that
\[
	\begin{gathered}
	|\Delta_{s-1}(\tilde{I}_w)\cap Rz|+|\Delta_{s-1}(\tilde{I}_w)\cap R|=|\Delta_{s-1}(\tilde{I}_w)|=n,\\
	|\Delta_s(\tilde{I}_w)\cap R|+|\Delta_s(\tilde{I}_w)\cap Rz|=|\Delta_s(\tilde{I}_w)|=n.
	\end{gathered}
\]

\begin{lemma}
For $w\neq \gw_s$,
\[
	|\Delta_{s-1}(\tilde{I}_w)\cap Rz|\ge n-|\Delta(J_e\phi_s)\cap\Delta(J)|.
\]
\end{lemma}

\begin{proof}	
Observe that $J_e(z-(\gw_s-w)\phi_s-\mu^{(s-1)})\subset\tilde{I}_w$ and $J\subset\tilde{I}_w$. Therefore $\Sigma(J_e\phi_s)\cup\Sigma(J)\subset\Sigma_{s-1}(\tilde{I}_w)\cap R$, that is 
\[
	\Delta_{s-1}(\tilde{I}_w)\cap R\subset\Delta(J_e\phi_s)\cap\Delta(J).
\]
Hence $|\Delta_{s-1}(\tilde{I}_w)\cap R|\le |\Delta(J_e\phi_s)\cap\Delta(J)|$, equivalent to the second equality.
\end{proof}

\begin{lemma}
$|\Delta(J_e\phi_s)|=\wt(e)+s$.
\end{lemma}

\begin{proof}
Note that
\[
	\begin{split}
	|\Delta(J_e\phi_s)|&=|\Sigma(R)\backslash\Sigma(J_e\phi_s)|=|\Delta(J_e)|+|\Sigma(R)\backslash\Sigma(R\phi_s)|\\
	&=\wt(e)+|S\backslash(s+S)|=\wt(e)+s.
	\end{split}
\]
The equality $|S\backslash(s+S)|=s$ holds for any numerical semigroup and can be proved by induction on the Frobenius number.
\end{proof}

\begin{theorem}\label{mxmjw}
The value $|\Delta_{s-1}(\tilde{I}_w)\cap Rz|$ is smallest for $w=\gw_s$, provided that
\[
	|\Delta(J)\cup\Delta(R\phi_s)|-s>2\wt(e).
\]
\end{theorem}

\begin{proof}
We need to show that for $w\neq\gw_s$,
\[
	|\Delta_{s-1}(\tilde{I}_w)\cap Rz|>|\Delta_{s-1}(\tilde{I}_{\gw_s})\cap Rz|.
\]
By \eqref{xjakd} and the previous lemmas, a sufficient condition for the above is
\[
	\begin{split}
	&n-|\Delta(J)\cap\Delta(J_e\phi_s)|>\wt(e)\\
	&\iff n-|\Delta(J)|-|\Delta(J_e\phi_s)|+|\Delta(J)\cup\Delta(J_e\phi_s)|>\wt(e)\\
	&\iff |\Delta(J)\cup\Delta(J_e\phi_s)|-s>2\wt(e)
	\end{split}	
\]
since $|\Delta(J)|=n$. Finally note that $|\Delta(J)\cup\Delta(J_e\phi_s)| \ge |\Delta(J)\cup\Delta(R\phi_s)|$.
\end{proof}

Note that $|\Delta_{s-1}(\tilde{I}_w)\cap Rz|$ is smallest when so is
\[
	\begin{split}
	|\Delta_{s-1}(\tilde{I}_w)\cap Rz|-|\Delta_s(\tilde{I}_w)\cap Rz|
	&=|\Delta_s(\tilde{I}_w)\cap R|-|\Delta_{s-1}(\tilde{I}_w)\cap R|\\
	&=|(\Delta_s(\tilde{I}_w)\backslash\Delta_{s-1}(\tilde{I}_w))\cap R)|\\
	&=|\Sigma_{s-1}(\tilde{I}_w)\cap\Delta_s(\tilde{I}_w)\cap R|.
	\end{split}
\]
since $|\Delta_s(\tilde{I}_w)\cap Rz|=|\Delta_s(I_{v^{(s)}})\cap Rz|$ is independent of $w$. The value
\[
	|\Sigma_{s-1}(\tilde{I}_w)\cap\Delta_s(\tilde{I}_w)\cap R|
\]
can be computed using the Gr\"obner bases of $\tilde{I}_w$ with respect to $>_s$ and $>_{s-1}$. As we saw in Section \ref{jksadde}, the Gr\"obner basis of $\tilde{I}_w$ with respect to $>_{s-1}$ is determined from $\tilde{B}$, the Gr\"obner bases of $\tilde{I}_w$ with respect to $>_s$. Precisely, according to Proposition \ref{qhkxa}, the set 
\[
	\Sigma_{s-1}(\tilde{I}_w)\cap\Delta_s(\tilde{I}_w)\cap R
\]
is determined by $\LM_{s-1}(F_j(z+w\phi_s))$ that lies in $\Delta_s(\tilde{I}_w)\cap R$. We note that for each $j$, there is a unique $w_j\in\F$ such that
\[
	\LM_{s-1}(F_j(z+w_j\phi_s))\in Rz,
\]
and $\LM_{s-1}(F_j(z+w\phi_s))=\LM(F_j^U\phi_s)\in R$ if and only if $w\neq w_j$. In fact, 
\[
	w_j=-\frac{d}{\LC(F_j^U)},
\]
where $d$ is the coefficient of the monomial $\LM(F_j^U\phi_s)$ in $F_j^D$.

\begin{proposition}\label{kaqqe}
\[
	\begin{split}
	\Sigma_{s-1}(\tilde{I}_w)\cap\Delta_s(\tilde{I}_w)\cap R
	&=\bigcup_{w_j\neq w}\Sigma_{s-1}(F_j(z+w\phi_s))\cap\Delta_s(\tilde{I}_w)\\
	&=\bigsqcup_{c\neq w}\bigcup_{w_j=c}\Sigma_{s-1}(F_j(z+w\phi_s))\cap\Delta_s(\tilde{I}_w)	
	\end{split}
\]
where $\sqcup$ denotes disjoint union.
\end{proposition}

\begin{proof}
The first equality follows from Proposition \ref{qhkxa}. It remains to show that the second union is disjoint. Assume that for $c_1\neq c_2$, there is a monomial $\phi\in R$ such that $\phi$ is in the intersection of
\[
	\bigcup_{w_j=c_1}\Sigma_{s-1}(F_j(z+w\phi_s))\cap\Delta_s(\tilde{I}_w)
\]
and
\[
	\bigcup_{w_j=c_2}\Sigma_{s-1}(F_j(z+w\phi_s))\cap\Delta_s(\tilde{I}_w).
\]
Let
\[
	\phi=\psi \LM_{s-1}(F_{j_1}(z+w\phi_s))
	=\chi \LM_{s-1}(F_{j_2}(z+w\phi_s))
\]
with $w_{j_1}=c_1$, $w_{j_2}=c_2$, and monomials $\psi$, $\chi$. Then we will show that
\begin{equation}\label{cjqdd}
	\LM_s(\frac{\psi}{\LC(F_{j_1}^U)}F_{j_1}(z+w\phi_s)-\frac{\chi}{\LC(F_{j_2}^U)} F_{j_2}(z+w\phi_s))=\phi,
\end{equation}
contradicting the assumption that $\phi\in\Delta_s(\tilde{I}_w)$. Indeed notice that $\phi=\LM(\psi F_{j_1}^U\phi_s)=\LM(\chi F_{j_2}^U\phi_s)$. Hence the coefficient of the monomial $\phi$ in the first term of the polynomial in \eqref{cjqdd} is
\[
	\frac{1}{\LC(F_{j_1}^U)}(w+d_1)
\]
where $d_1$ is the coefficient of the monomial $\LM( F_{j_1}^U\phi_s)$ in $F_{j_1}^D$. In the same way, the coefficient of the monomial $\phi$ in the second term after the minus in \eqref{cjqdd} is
\[
	\frac{1}{\LC(F_{j_2}^U)}(w+d_2)
\]
where $d_2$ is the coefficient of the monomial $\LM(F_{j_2}^U\phi_s)$ in $F_{j_2}^D$. These two coefficients are different because we assumed 
\[
	w_{j_1}=-\frac{d_1}{\LC(F_{j_1}^U)}\neq w_{j_2}=-\frac{d_2}{\LC(F_{j_2}^U)}.
\]
Hence \eqref{cjqdd} follows.
\end{proof}

We observe that for $c,w\in\F$ with $w\neq c$, 
\[
	\bigcup_{w_j=c}\Sigma_{s-1}(F_j(z+w\phi_s))\cap\Delta_s(\tilde{I}_w)
	=\bigcup_{w_j=c}\Sigma(F_j^U\phi_s)\cap\Delta(\set{G_i^D}).
\]
Therefore this set is independent of $w$, and is determined by $B^{(s)}$. Let 
\[
	d_c=\Bigl|\bigcup_{w_j=c}\Sigma(F_j^U\phi_s)\cap\Delta(\set{G_i^D})\Bigr|.
\]
Then Proposition \ref{kaqqe} implies 
\[
	|\Delta_{s-1}(\tilde{I}_w)\cap Rz|-|\Delta_s(\tilde{I}_w)\cap Rz|=\sum_{c\neq w}d_c
\]
is smallest when $w=c$ with $d_c$ largest. Now we can elaborate the main steps of the interpolation decoding algorithm as follows:

\begin{description}
\item[M1] If $s$ is a nongap $\le u$, then do the following, but otherwise let $\tilde{B}=\set{G_i,F_j}$. 
\begin{description}
\item[M1.1] Compute the set $W=\set{w_j}$, where 
\[
	w_j=-\frac{d}{\LC(F_j^U)},
\]
and $d$ is the coefficient of the monomial $\LM(F_j^U\phi_s)$ in $F_j^D$.
\item[M1.2] For each $c\in W$, compute the value 
\[
	d_c=\Bigl|\bigcup_{w_j=c}\Sigma(F_j^U\phi_s)\cap\Delta(\set{G_i^D})\Bigr|.
\]
\item[M1.3] Let $w^{(s)}=c$ with largest $d_c$, and let 
\[
	\tilde{B}=\set{G_{i}(z+w^{(s)}\phi_s),F_j(z+w^{(s)}\phi_s)}.
\]
\end{description}
\item[M2] Suppose $\tilde{B}=\set{\tilde{G}_i,\tilde{F}_j}$. Let
\[
	B^{(s-1)}=\set{\tilde{G}_i,\tilde{F}_j\mid \LM_{s-1}(\tilde{F}_j)\in R\cap\Delta_s(\set{\tilde{G}_i})}'\cup\bigcup_{j}'\spoly(\tilde{F}_j).
\]
\end{description}

\begin{theorem}
The algorithm outputs $w^{(s)}=\gw_s$ for all $s\in S,s\le u$ if
\[
	d_u=\min_{s\in S,s\le u}\nu(s)>2\wt(e),
\]
where $\nu(s)=|\Delta(J)\cup\Delta(R\phi_s)|-s$ for $s\in S$. Moreover $d_u\ge n-u$.
\end{theorem}

\begin{proof}
By Theorem \ref{mxmjw}, the condition $d_u>2\wt(e)$ implies that the algorithm computes $w^{(s)}=\gw_s$ for each iteration for nongap $s$ from $u$ to $0$. To see $d_u\ge n-u$, notice that $|\Delta(J)\cup\Delta(R\phi_s)|\ge |\Delta(J)|=n$.
\end{proof}

\section{Decoding Hermitian Codes}\label{ckdwdd}

In this section, we demonstrate the decoding algorithm on the Hermitian codes defined on Hermitian curves with equation 
\[
	y^q+y-x^{q+1}=0
\]
over $\F_{q^2}$. There are $q^3$ rational points on the Hermitian curve, and $J=\ideal{x^{q^2}-x}$. We now determine the performance of the decoding algorithm for the Hermitian code $C_u$.

\begin{theorem}
For nongap $u<q^3$,
\[
	\begin{aligned}
	d_u=\begin{cases}
	q^3-aq & b \le a+q-q^2\\
	q^3-u & b> a+q-q^2
	\end{cases}
	\end{aligned}
\]
if $u=aq+b$, $0\le b<q$.
\end{theorem}

\begin{proof}
We first compute $\nu(s)$ for nongap $s=qs_1+s_2<q^3$. As 
\[
	\begin{split}
	|\Delta(J)\cup\Delta(R\phi_s)|&=|\Sigma(J)\cap\Delta(R\phi_s)|+|\Delta(J)|\\
	&=|\set{t\in S\mid q^3+t\notin s+S}|+q^3.
	\end{split}
\]
we have $\nu(s)=|\set{t\in S\mid q^3+t-s\notin S}|+q^3-s$. Note that 
\[
	q^3+t-s=q(q^2+t_1-s_1)+t_2-s_2
\]
with $t=qt_1+t_2$. Therefore $q^3+t-s\notin S$ if and only if  
\[
	t_2-s_2\ge 0,\quad q^2+t_1-s_1<t_2-s_2
\]
or 
\[
	t_2-s_2<0,\quad q^2+t_1-s_1<q+1+t_2-s_2.
\]
The first case is actually impossible since $s_1<q^2$. Hence
\[
	|\set{t\in S\mid q^3+t-s\notin S}|=s_2\max\set{s_1-s_2+q+1-q^2,0}.
\]
Thus 
\[
	\nu(s)=s_2\max\set{s_1-s_2+q+1-q^2,0}+q^3-s
\]
for $s=qs_1+s_2<q^3$. If $a-b+q-q^2\ge 0$, then the minimum is attained at $s=aq$, and hence
$d_u=q^3-aq$ while if $a-b+q-q^2<0$, then the minimum is attained at $s=u$, and hence $d_u=q^3-u$.
\end{proof}

\begin{figure}
\begin{center}
\begin{tikzpicture}[scale=2,x=3pt,y=2pt,
	declare function={
		q = 3;
		b(\u) = mod(\u,q);
		a(\u) = div(\u,q);
		d(\u) = b(\u) <= a(\u) - (q^2 - q) ? q^3 - a(\u)*q : q^3 - \u;
		s(\u)= a(\u) - b(\u) < 0 ? a(\u)*q+a(\u) : \u;
		del(\u)= d(s(\u)); 
	}]	
	\draw[samples at={0,1,...,27},mark=ball,mark size=.7pt,ball color=black,only marks,variable=\u] 
		plot (\u,{del(\u)});	
	\draw (0,0) -- (28,0);
	\draw (0,0) -- (0,27);
	\foreach \x in {0,10,...,27}
		\draw (\x,0) -- (\x,3pt);	
	\foreach \x in {0,5,...,27}
		\draw (\x,0) -- (\x,2pt);		
	\foreach \x in {10,20,...,27}
		\draw (\x,0) node[anchor=north] {\x};	
	\foreach \y in {0,10,...,27}
		\draw (0,\y) -- (3pt,\y);
	\foreach \y in {0,5,...,27}
		\draw (0,\y) -- (2pt,\y);
	\foreach \y in {0,10,...,27}
		\draw (0,\y) node[anchor=east] {\y};			
\end{tikzpicture}
\caption{Decoding performance of Hermitian codes of length $27$}\label{ckqe}
\label{default}
\end{center}
\end{figure}

Figure \ref{ckqe} shows the decoding performance $d_u$ of $C_u$ over the Hermitian curve $y^3+y-x^4=0$ over $\F_9$, where $\F_9=\F_3(\ga)$ with $\ga^2-\ga-1=0$. The $27$ rational points on the curve are
\[
	\begin{gathered}
	( 0, 0 ), ( 0, \ga^2 ), ( 0, \ga^6 ), ( 1, 2 ), ( 1, \ga ), ( 1, \ga^3 ), ( 2, 2 ), ( 2, \ga ), ( 2, \ga^3 ),\\
 	( \ga, 1 ), ( \ga, \ga^7 ), ( \ga, \ga^5 ), ( \ga^2, 2 ), ( \ga^2, \ga ), ( \ga^2, \ga^3 ), ( \ga^7, 1 ), ( \ga^7, \ga^7 ), ( \ga^7, \ga^5 ),\\
	 ( \ga^5, 1 ), ( \ga^5, \ga^7 ), ( \ga^5, \ga^5 ), ( \ga^3, 1 ), ( \ga^3, \ga^7 ), ( \ga^3, \ga^5 ), ( \ga^6, 2 ), ( \ga^6, \ga ), ( \ga^6, \ga^3 ),
	\end{gathered}
\]
and there is a unique point $P_\infty$ at infinity. As $\gd(x)=3$ and $\gd(y)=4$, the numerical semigroup of the coordinate ring $R$ is
\[
	S=\tuple{3,4}=\set{0,3,4,6,7,8,9,10,\dots}.
\] 
Note that $S$ has three gaps $1$, $2$, and $5$. The monomials of $R$ correspond to nongaps in $S$ and are displayed in the diagram
\[
\begin{tikzpicture}[scale=1.2,x=4mm,y=4mm,baseline=(current bounding box)]
  	\draw[dimgray,fill] (0,0) |- (11,3) |- (0,0);	
	\draw (0,0) grid[step=4mm] +(11,3);
	\draw (0,0) +(.5,.5) node[scale=.6] {$1^{\vphantom{1}}$};
	\draw (1,0) +(.5,.5) node[scale=.6] {$x^{\vphantom{1}}$};
	\draw (10,0) +(.5,.5) node[scale=.6] {$\cdots$};
	\foreach \x in {2,...,9}	  
	  		\draw (\x,0) +(.5,.5) node[scale=.6] {$x^{\x}$};
	\draw (0,1) +(.5,.5) node[scale=.6] {$y^{\vphantom{1}}$};
	\draw (1,1) +(.5,.5) node[scale=.6] {$xy^{\vphantom{1}}$};
	\draw (10,1) +(.5,.5) node[scale=.6] {$\cdots$};
	\foreach \x in {2,...,9}	  
	  		\draw (\x,1) +(.5,.5) node[scale=.6] {$x^{\x}y$};
	\draw (0,2) +(.5,.5) node[scale=.6] {$y^2$};
	\draw (1,2) +(.5,.5) node[scale=.6] {$xy^2$};
	\draw (10,2) +(.5,.5) node[scale=.6] {$\cdots$};
	\foreach \x in {2,...,9}	  
	  	\draw (\x,2) +(.5,.5) node[scale=.6] {$x^{\x}y^2$};	  
\end{tikzpicture}
\]

Let $u=16$. Then the Hermitian code $C_{16}$ has dimension $14$ and minimum distance $11$, and the decoding algorithm can correct up to $5$ errors. 
Suppose we received the vector
\[
	v=(0,0,0,0,0,\ga^2,2,0, 0, 0, 0, 0, 0, 0, 0, 0, 0, 0, 0, \ga^3, 0, 0, \ga^7, 0, 0, 2, 0)
\]
from a noisy channel. We now follow the steps of the decoding algorithm.

The algorithm first compute the Lagrange interpolation of $v$,
\[
	\begin{split}
	h_v&=\ga^3x^8y^2 + x^7y^2 + \ga^6x^8y + \ga^7x^6y^2 + x^7y + 2x^8 + \ga^5x^5y^2 \\
	&\quad + x^6y + \ga x^7 + \ga^3x^4y^2 + \ga x^5y + \ga^6x^6 + \ga^6x^3y^2 + \ga^6x^4y\\
	&\quad + 2x^2y^2 + \ga^7x^3y + 2x^4 + \ga^2xy^2 + \ga^2x^3 + \ga^3xy + x.
	\end{split}
\]
The algorithm iterates the main steps for $s$ from $N=\gd(h_v)=32$ to $0$. The ideal $J$ has Gr\"obner basis $\set{\eta_1=x^9-x}$. Hence the Gr\"obner basis of $I_{v^{(32)}}=I_v$ is
\[
B^{(32)}=\left\{\begin{array}{rcrcl}
G_1&=&                          0\,z&+&x^9-x\\
F_1&=&                           1\,z&+&\ga^7x^8y^2+\cdots\\
\end{array}\right\}
\]
\begin{center}
\begin{tikzpicture}[scale=1.2,x=4mm,y=4mm,baseline=(current bounding box)]
  \draw[dimgray,fill] (0,0) |- (11,3) |- (0,0);
  \draw[lightgray,fill] (11,3) |- (0,0) |- (11,3);  
  \draw (0,0) grid[step=4mm] +(11,3);  
  \draw (0.5,0.5) node[scale=.6] {$1$};
\end{tikzpicture}
\quad
\begin{tikzpicture}[scale=1.2,x=4mm,y=4mm,baseline=(current bounding box)]
  \draw[dimgray,fill] (0,0) |- (11,3) |- (0,0);
  \draw[lightgray,fill] (11,3) |- (9,0) |- (11,3);
  \draw (0,0) grid[step=4mm] +(11,3);  
  \draw (9.5,0.5) node[scale=.6] {$x^9$}; 
\end{tikzpicture}
\end{center}
Here the left diagram exhibits the monomials in $\Sigma_{32}(I_{v^{(32)}})\cap Rz$ omitting the common $z$ variable, while the right diagram shows the monomials in $\Sigma_{32}(I_{v^{(32)}})\cap R$. The leading terms of the polynomials in the Gr\"obner basis are also shown.

For $s\ge u=16$ or a gap $s$, as $\tilde{B}=\set{\tilde{G}_i,\tilde{F}_j}=B^{(s)}=\set{G_i,F_j}$ in the step \textbf{M1}, we will omit the tilde in the following. In the step \textbf{M2}, $\LM_{31}(F_1)=x^8y^2\in R\cap \Delta_{32}(G_1)$, and the lcms of $\LM_{31}(F_1)=x^8y^2$ and $\LM_{32}(G_1)=x^9$ are $x^9y^2$ and $x^{12}$. Hence $\spoly(F_1)=\set{\ga xz+\ga^5x^8y^2 + \dots+ \ga^5x^2,\ga yz+\ga^5x^{11} + \dots + \ga^2xy}$. Then the Gr\"obner basis of $I_{v^{(31)}}$ is
\[
B^{(31)}=\left\{\begin{array}{rcrcl}
G_1&=&                          0\,z&+&x^9+\cdots\\
G_2&=&                          1\,z&+&\ga^7x^8y^2+\cdots\\
F_1&=&                     \ga x\,z&+&\ga^5x^8y^2+\cdots\\
F_2&=&                      \ga y\,z&+&\ga^5x^{11}+\cdots\\
\end{array}\right\}
\]
\begin{center}
\begin{tikzpicture}[scale=1.2,x=4mm,y=4mm,baseline=(current bounding box)]
  \draw[dimgray,fill] (0,0) |- (11,3) |- (0,0);
  \draw[lightgray,fill] (11,3) |- (1,0) |- (0,1) |- (11,3);    
  \draw (0,0) grid[step=4mm] +(11,3);  
  \draw (1.5,0.5) node[scale=.6] {$x$};
  \draw (0.5,1.5) node[scale=.6] {$y$};  
\end{tikzpicture}
\quad
\begin{tikzpicture}[scale=1.2,x=4mm,y=4mm,baseline=(current bounding box)]
  \draw[dimgray,fill] (0,0) |- (11,3) |- (0,0);
  \draw[lightgray,fill] (11,3) |- (9,0) |- (8,2) |- (11,3);
  \draw (0,0) grid[step=4mm] +(11,3);  
  \draw (9.5,0.5) node[scale=.6] {$x^9$};
  \draw (8.5,2.5) node[scale=.6] {$x^8y^2$};  
\end{tikzpicture}
\end{center}
As $\LM_s(F_1),\LM_s(F_2)\in Rz$ for  $s=31,30$, there is no change in the Gr\"obner basis. So we get to the unaltered Gr\"obner basis of $I_{v^{(29)}}$ 
\[
B^{(29)}=\left\{\begin{array}{rcrcl}
G_1&=&                          0\,z&+&x^9+\cdots\\
G_2&=&                          1\,z&+&\ga^7x^8y^2+\cdots\\
F_1&=&                     \ga x\,z&+&\ga^5x^8y^2+\cdots\\
F_2&=&                     \ga y\,z&+&\ga^5x^{11}+\cdots\\
\end{array}\right\}
\]
\begin{center}
\begin{tikzpicture}[scale=1.2,x=4mm,y=4mm,baseline=(current bounding box)]
  \draw[dimgray,fill] (0,0) |- (11,3) |- (0,0);
  \draw[lightgray,fill] (11,3) |- (1,0) |- (0,1) |- (11,3);    
  \draw (0,0) grid[step=4mm] +(11,3);  
  \draw (1,0) +(.5,.5) node[scale=.6] {$x$};
  \draw (0,1) +(.5,.5) node[scale=.6] {$y$};  
\end{tikzpicture}
\quad
\begin{tikzpicture}[scale=1.2,x=4mm,y=4mm,baseline=(current bounding box)]
  \draw[dimgray,fill] (0,0) |- (11,3) |- (0,0);
  \draw[lightgray,fill] (11,3) |- (9,0) |- (8,2) |- (0,3) |- (11,3);
  \draw (0,0) grid[step=4mm] +(11,3);  
  \draw (9,0) +(.5,.5) node[scale=.6] {$x^9$};
  \draw (8,2) +(.5,.5) node[scale=.6] {$x^8y^2$};  
\end{tikzpicture}
\end{center}
Now since $\LM_{29}(G_2)=x^8y^2$ divides $\LM_{28}(F_1)=x^8y^2$ and $\LM_{29}(G_1)=x^9$ divides $\LM_{28}(F_2)=x^{11}$, both $\LM_{28}(F_1)$ and $\LM_{28}(F_2)$ are in $R\cap\Sigma_{29}(G_1,G_2)$ Thus
\[
	\begin{aligned}
	\spoly(F_1)&=\set{2xz + \ga^5z+\ga^6x^9y + \dots + \ga x},\\
	\spoly(F_2)&=\set{2yz+\ga^6x^8y^2 + \dots+ \ga^5xy}.
	\end{aligned}
\]
Hence the Gr\"obner basis of $I_{v^{(28)}}$ is
\[
B^{(28)}=\left\{\begin{array}{rcrcl}
G_1&=&                          0\,z&+&x^9+\cdots\\
G_2&=&                          1\,z&+&\ga^7x^8y^2+\cdots\\
F_1&=&                     (2 x+\ga^5)\,z&+&\ga^6 x^9y+\cdots\\
F_2&=&                     2 y\,z&+&\ga^6 x^8y^2+\cdots\\
\end{array}\right\}
\]
\begin{center}
\begin{tikzpicture}[scale=1.2,x=4mm,y=4mm,baseline=(current bounding box)]
  \draw[dimgray,fill] (0,0) |- (11,3) |- (0,0);
  \draw[lightgray,fill] (11,3) |- (1,0) |- (0,1) |- (11,3);    
  \draw (0,0) grid[step=4mm] +(11,3);  
  \draw (1,0) +(.5,.5) node[scale=.6] {$x$};
  \draw (0,1) +(.5,.5) node[scale=.6] {$y$};  
\end{tikzpicture}
\quad
\begin{tikzpicture}[scale=1.2,x=4mm,y=4mm,baseline=(current bounding box)]
  \draw[dimgray,fill] (0,0) |- (11,3) |- (0,0);
  \draw[lightgray,fill] (11,3) |- (9,0) |- (8,2) |- (0,3) |- (11,3);
  \draw (0,0) grid[step=4mm] +(11,3);  
  \draw (9,0) +(.5,.5) node[scale=.6] {$x^9$};
  \draw (8,2) +(.5,.5) node[scale=.6] {$x^8y^2$};  
\end{tikzpicture}
\end{center}
Similar steps are iterated. Eventually, we get to the Gr\"obner basis of $I_{v^{(16)}}$,
\[
B^{(16)}=\left\{\begin{array}{rcrcl}
G_1&=&                          0\,z&+&x^9+\cdots\\
G_2&=&                     (\ga^2 xy+\dots)\,z&+&\ga^2 x^7y+\cdots\\
F_1&=&(\ga^2x^2+\cdots)\,z&+&0\\
F_2&=&            (\ga^5y^2+\cdots)\,z&+&x^8+\cdots\\
\end{array}\right\}
\]
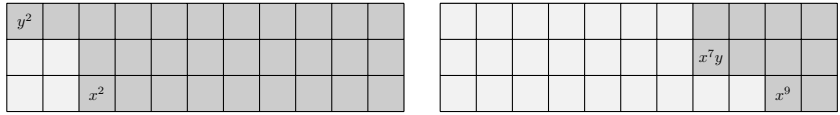
\begin{center}
\begin{tikzpicture}[scale=1.2,x=4mm,y=4mm,baseline=(current bounding box)]
  \draw[dimgray,fill] (0,0) |- (11,3) |- (0,0);
  \draw[lightgray,fill] (11,3) |- (2,0) |- (0,2) |- (11,3);    
  \draw (0,0) grid[step=4mm] +(11,3);  
  \draw (2,0) +(.5,.5) node[scale=.6] {$x^2$};
  \draw (0,2) +(.5,.5) node[scale=.6] {$y^2$};  
\end{tikzpicture}
\quad
\begin{tikzpicture}[scale=1.2,x=4mm,y=4mm,baseline=(current bounding box)]
  \draw[dimgray,fill] (0,0) |- (11,3) |- (0,0);
  \draw[lightgray,fill] (11,3) |- (9,0) |- (7,1) |- (0,3) |- (11,3);
  \draw (0,0) grid[step=4mm] +(11,3);  
  \draw (9,0) +(.5,.5) node[scale=.6] {$x^9$};
  \draw (7,1) +(.5,.5) node[scale=.6] {$x^7y$};  
\end{tikzpicture}
\end{center}
Now $s=16$ is a nongap  and $\le u=16$. So in the step \textbf{M1}, we proceed to guess $\gw_{16}$ for the monomial $\phi_{16}=x^4y$. The leading coefficient of $F_1$ is $\ga^2$ and the coefficient of the monomial $x^6y$ in $F_1$ is $0$, where $x^6y$ is the leading monomial of $x^2\phi_{16}$. Hence $w_1=-(0/\ga^2)=0$. The leading coefficient of $F_2$ is $\ga^5$ and the coefficient of the monomial $x^8$ in $F_2$ is $1$, where $x^8$ is the leading monomial of $y^2\phi_{16}$. Hence $w_2=-(1/\ga^5)=\ga^7$. So $W=\set{0,\ga^7}$. The shape of 
\[
	\bigcup_{w_j=0}\Sigma(F_j^U\phi_{16})\cap\Delta(\set{G_i^D})=\Sigma(x^6y)\cap\Delta(x^9,x^7y)
\]
is
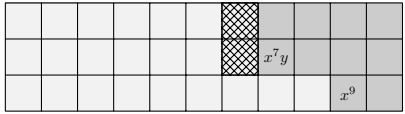
\[
\begin{tikzpicture}[scale=1.2,x=4mm,y=4mm,baseline=(current bounding box)]
  \draw[dimgray,fill] (0,0) |- (11,3) |- (0,0);
  \draw[lightgray,fill] (11,3) |- (9,0) |- (7,1) |- (0,3) |- (11,3);
  \draw[pattern=crosshatch] (6,1) |- (7,3) |- (6,1);  
  \draw (0,0) grid[step=4mm] +(11,3);  
  \draw (9,0) +(.5,.5) node[scale=.6] {$x^9$}; 
  \draw (7,1) +(.5,.5) node[scale=.6] {$x^7y$};    
\end{tikzpicture} 
\]
and thus $d_0=2$. On the other hand, the shape of
\[
	\bigcup_{w_j=\ga^7}\Sigma(F_j^U\phi_{16})\cap\Delta(\set{G_i^D})=\Sigma(x^8)\cap\Delta(x^9,x^7y)
\]
is
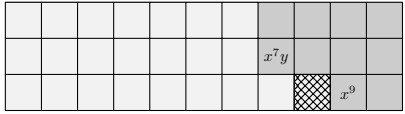
\[
\begin{tikzpicture}[scale=1.2,x=4mm,y=4mm,baseline=(current bounding box)]
  \draw[dimgray,fill] (0,0) |- (11,3) |- (0,0);
  \draw[lightgray,fill] (11,3) |- (9,0) |- (7,1) |- (0,3) |- (11,3);
  \draw[pattern=crosshatch] (8,0) |- (9,1) |- (8,0);  
  \draw (0,0) grid[step=4mm] +(11,3);  
  \draw (9,0) +(.5,.5) node[scale=.6] {$x^9$}; 
  \draw (7,1) +(.5,.5) node[scale=.6] {$x^7y$};    
\end{tikzpicture}
\]
and thus $d_{\ga^7}=1$. Hence we take $w^{(16)}=0$. Then
\[
\tilde{B}=\left\{\begin{array}{rcrcl}
\tilde{G}_1&=&                          0\,z&+&x^9+\cdots\\
\tilde{G}_2&=&                     (\ga^2 xy+\dots)\,z&+&\ga^2 x^7y+\cdots\\
\tilde{F}_1&=&(\ga^2x^2+\cdots)\,z&+&0\\
\tilde{F}_2&=&            (\ga^5y^2+\cdots)\,z&+&x^8+\cdots\\
\end{array}\right\}
\]
In the step \textbf{M2}, $\LM_{15}(\tilde{F}_1)=x^2z\in Rz$ and $\LM_{15}(\tilde{F}_2)=x^8\in R\cap\Delta_{15}(\tilde{G}_1,\tilde{G}_2)$. So $\spoly(\tilde{F}_1)=\set{\tilde{F}_1}$, and since the lcm of $\LM_{15}(\tilde{F}_2)=x^8$ and $\LM_{15}(\tilde{G}_1)=x^9$ is $x^9$, and the lcms of $\LM_{15}(\tilde{F}_2)$ and $\LM_{15}(\tilde{G}_2)=x^7y$ are $x^8y$ and $x^{11}$,
\[
	\spoly(\tilde{F}_2)=\set{\ga^5xy^2z+\dots+x,\ga^5x^4z+\dots+x^2,\ga^5x^3y^2z+\dots+\ga^5xy}.
\]
Removing redundant polynomials, we have 
\[
	\begin{gathered}
	\set{\tilde{G}_1,\tilde{G}_2,\tilde{F}_2}'=\set{\tilde{G}_2,\tilde{F}_2},\\
	(\spoly(\tilde{F}_1)\cup\spoly{\tilde{F}_2})'=\set{\tilde{F}_1,\ga^5xy^2z+\dots+x}.
	\end{gathered}
\]
Thus the Gr\"obner basis of $I_{v^{(15)}}$ is
\[
B^{(15)}=\left\{\begin{array}{rcrcl}
G_1&=&(\ga^2 xy+\cdots)\,z&+&\ga^2x^7y+\cdots\\
G_2&=&(\ga^5y^2+\cdots)\,z&+&x^8+\cdots\\
F_1&=&(\ga^2x^2+\cdots)\,z&+&0\\
F_2&=&            (\ga^5xy^2+\cdots)\,z&+&2 x^7y+\cdots\\
\end{array}\right\}
\]
\begin{center}
\begin{tikzpicture}[scale=1.2,x=4mm,y=4mm,baseline=(current bounding box)]
  \draw[dimgray,fill] (0,0) |- (11,3) |- (0,0);
  \draw[lightgray,fill] (11,3) |- (2,0) |- (1,2) |- (0,3) |- (11,3);  
  \draw (0,0) grid[step=4mm] +(11,3);  
  \draw (2,0) +(.5,.5) node[scale=.6] {$x^2$};
  \draw (1,2) +(.5,.5) node[scale=.6] {$xy^2$};     
\end{tikzpicture}
\quad
\begin{tikzpicture}[scale=1.2,x=4mm,y=4mm,baseline=(current bounding box)]
  \draw[dimgray,fill] (0,0) |- (11,3) |- (0,0);
  \draw[lightgray,fill] (11,3) |- (8,0) |- (7,1) |- (0,3) |- (11,3);
  \draw (0,0) grid[step=4mm] +(11,3);  
  \draw (8,0) +(.5,.5) node[scale=.6] {$x^8$}; 
  \draw (7,1) +(.5,.5) node[scale=.6] {$x^7y$};  
\end{tikzpicture}
\end{center}
Continuing in this way, after the last iteration for $s=0$, we get to the Gr\"obner basis of $I_{v^{(-1)}}$,
\[
B^{(-1)}=\left\{\begin{array}{rcrcl}
G_1&=&(\ga^2xy+\cdots)\,z&+&\ga^2x^7y+\cdots\\
G_2&=&(\ga^5y^2+\cdots)\,z&+&x^8+\cdots\\
F_1&=&(\ga^2x^2+\cdots)\,z&+&0\\
F_2&=&            (xy^2+\cdots)\,z&+&0\\
\end{array}\right\}
\]
\begin{center}
\begin{tikzpicture}[scale=1.2,x=4mm,y=4mm,baseline=(current bounding box)]
  \draw[dimgray,fill] (0,0) |- (11,3) |- (0,0);
  \draw[lightgray,fill] (11,3) |- (2,0) |- (1,2) |- (0,3) |- (11,3);  
  \draw (0,0) grid[step=4mm] +(11,3);  
  \draw (2,0) +(.5,.5) node[scale=.6] {$x^2$};
  \draw (1,2) +(.5,.5) node[scale=.6] {$xy^2$};     
\end{tikzpicture}
\quad
\begin{tikzpicture}[scale=1.2,x=4mm,y=4mm,baseline=(current bounding box)]
  \draw[dimgray,fill] (0,0) |- (11,3) |- (0,0);
  \draw[lightgray,fill] (11,3) |- (8,0) |- (7,1) |- (0,3) |- (11,3);
  \draw (0,0) grid[step=4mm] +(11,3);  
  \draw (8,0) +(.5,.5) node[scale=.6] {$x^8$}; 
  \draw (7,1) +(.5,.5) node[scale=.6] {$x^7y$};  
\end{tikzpicture}
\end{center}
Finally, the algorithm output $(w^{(s)}=0\mid\text{nongap $s\le 16$})$.

\section{Remarks}

We presented an interpolation-based unique decoding algorithm for the primal algebraic geometry codes. The algorithm iteratively finds each entry of the sent message by a majority voting procedure. We showed that the majority voting is successful if the codeword encoding the message is close enough to the received vector.

\end{document}